\documentclass[12pt]{article}

\usepackage{epsfig}

\usepackage{amssymb}
\usepackage{amsfonts}

\usepackage{color}
\def\be{\begin{equation}}
\def\ee{\end{equation}}
\def\ba{\begin{array}{c}}
\def\ea{\end{array}}

\newcommand{\bea}{\begin{eqnarray}}
\newcommand{\eea}{\end{eqnarray}}

\newcommand{\pkt}{\!\!\succ\,\,}
\newcommand{\kt}{\rangle}
\newcommand{\br}{\langle}

\newtheorem{thm}{Theorem}

\newtheorem{lemma}[thm]{Lemma}

\newenvironment{proof}{\noindent
 {\bf Proof.}}{\hfill$\square$\vspace{3mm}\endtrivlist}

\begin{document}

\begin{center}

{\Large \bf

Quantum phase transitions in nonhermitian harmonic oscillator

%
%Hilbert spaces of states of
%
%$\,{\cal PT}-$symmetric harmonic oscillator
%
%near exceptional points

}

\vspace{0.8cm}

  {\bf Miloslav Znojil}$^{1, 2}$

%
%
%\vspace{1mm} Nuclear Physics Institute of the CAS, Hlavn\'{\i} 130,
%250 68 \v{R}e\v{z}, Czech Republic
%
%

\end{center}

\vspace{10mm}

 $^{1}$ {Department of Physics, Faculty of
Science, University of Hradec Kr\'{a}lov\'{e}, Rokitansk\'{e}ho 62,
50003 Hradec Kr\'{a}lov\'{e},
 Czech Republic}

 $^{2}$ {The Czech Academy of Sciences,
 Nuclear Physics Institute,
 Hlavn\'{\i} 130,
250 68 \v{R}e\v{z}, Czech Republic, {e-mail: znojil@ujf.cas.cz}}
 %\footnote{{e-mail: znojil@ujf.cas.cz}}

\newpage

\section*{Abstract}

The Stone theorem requires that in a physical Hilbert space
${\cal H}$ the time-evolution of a stable quantum system is unitary if and
only if the corresponding Hamiltonian $H$ is self-adjoint.
Sometimes, a simpler picture of the evolution may be constructed in
a manifestly unphysical Hilbert space ${\cal K}$ in which $H$ is
nonhermitian but ${\cal PT}-$symmetric. In applications,
unfortunately, one only rarely succeeds in circumventing the key
technical obstacle which lies in the necessary reconstruction of the
physical Hilbert space ${\cal H}$. For a ${\cal PT}-$symmetric
version of the spiked harmonic oscillator we show that in the
dynamical regime of the unavoided level crossings such a
reconstruction of ${\cal H}$ becomes feasible and, moreover,
obtainable by non-numerical means. The general form of such a
reconstruction of ${\cal H}$ enables one to render every exceptional
unavoided-crossing point tractable as a genuine, phenomenologically
most appealing quantum-phase-transition instant.

\newpage

\section*{Introduction}

In the Carroll's book about Alice's adventures
one reads that
the ``Cheshire Cat appears in a tree'',
and then he ``disappears but
his grin remains behind to float on its own in the air''
\cite{Alice}.
In a fairly close parallel to the Carroll's story
(and to its occasional, time-to-time use in physics  \cite{Aliceb,Alicec}),
the appearance, more than twenty years ago \cite{BB,DB}, of
the imaginary cubic potential $V(x)={\rm i}x^3$
resembled
the Cheshire Cat of
quantum theory.
The study of the model
directed the community of quantum physicists
towards the wide
acceptance of the concept of parity-times-time-reversal symmetry
(${\cal PT}-$symmetry) in
quantum mechanics of unitary systems \cite{Carl}.
In this new formulation of the theory a simpler picture of the
evolution
may sometimes be constructed \cite{ali}.

In a continued parallel, ``the grin'' (i.e., the inspiring idea of
${\cal PT}-$symmetry) is still ``floating in the air'' at present
\cite{Christodoulides,Carlbook}. At the same time, the imaginary
cubic potential itself has repeatedly been shown to disappear from
the scene of physics because ``there is no quantum-mechanical
Hamiltonian associated with it'' \cite{2012}. An evidence is now
available that many nonhermitian, Cheshire-Cat-resembling models of
dynamics exhibit certain ``unexpected wild properties''
\cite{Viola}. Many people now believe that at least some of the
similar benchmark models ``are not equivalent to Hermitian models,
but that they rather form a separate model class with purely real
spectra'' \cite{Uweho}.

In our present paper we intend to weaken such a wave of scepticism.
We will show that there exist non-trivial nonhermitian quantum
systems living in infinite-dimensional Hilbert spaces in which,
after an appropriate formulation of the theory, even an innovative,
quantum-phase-transition-opening ``wild'' behavior can still be
given an entirely conventional, unitary-evolution interpretation and
explanation compatible with the dictum of standard textbooks.

The purpose will be served by the nonhermitian but ${\cal
PT}-$symmetric toy model (with the real spectrum) represented by the
ordinary differential Schr\"{o}dinger equation
 \be
 \left (-\,\frac{d^2}{dr^2} + \frac{G}{r^2}+r^2
  \right )
 \, \varphi_{}(r) =
 E_{} \, \varphi_{}(r)\,.
 \label{SEnobs}
 \ee
In conventional textbooks \cite{Messiah}
the model is used to describe the radial motion of a
particle in a
$D-$dimensional harmonic oscillator well
$V(\vec{r})=|\vec{r}|^2$.
In these textbooks
the authors also add a comment that
even when $G > -1/4$ is negative,
the quantum system remains stable,
in a remarkable contrast to its unstable classical
analogue.

The latter remarks do not exhaust the list of the remarkable features of model
(\ref{SEnobs}).
In 1999, in
a way inspired by Bender and Boettcher \cite{BB} we showed,
in Ref.~\cite{ptho},
that
under the same constraint $G > -1/4$
the model remains stable
even when it ceases to be Hermitian.
We proved, in particular, that
the complex shift
of the line of coordinates
 \be
 r\ \to \ r(x)=x-{\rm i} c \ \in \ \mathbb{C}\,,
 \ \ \ \ c>0\,,\ \ \
 x \in (-\infty,\infty)\,
 \label{shift}
 \ee
in the same ordinary differential equation (\ref{SEnobs}) still
keeps the energy spectrum real, discrete and bounded from below. The
regularization of the centrifugal-type singularity in the resulting
nonhermitian but ${\cal PT}-$symmetric (i.e.,
parity-times-time-reversal-symmetric \cite{Carl,cartoon})
Hamiltonian
 \be
 H^{(\alpha)}=
 -\,\frac{d^2}{dx^2} + (x-ic)^2
 + \frac{G}{(x-ic)^2}\,,
  \ \ \ \ \alpha=\sqrt{G+1/4}>0\,,
  \ \ \ \ c>0\,,
  \ \ \ \
 x \in (-\infty,\infty)
 \label{Halfa}
 \ee
made the model eligible
as an unusual but still exactly solvable example
in supersymmetric quantum mechanics \cite{susy}.

Mathematically, our Hamiltonian operator $H^{(\alpha)}$ is defined
in Hilbert space ${\cal K}=L^2(\mathbb{R})$ of square-integrable
functions of the new real variable $x$. In Ref.~\cite{ptho} it has
been shown that in spite of the manifest nonhermiticity of
Hamiltonian (\ref{Halfa}) in ${\cal K}$, its spectrum is all real
and defined, in terms of two quantum numbers, by compact formula
 \be
 E=E_{n}^{(Q)}=4n+2 - 2 Q \alpha\,,
 \ \ \ \ \ \ Q = \pm 1\,,
 \ \ \ \ \  n = 0, 1,
2, \ldots
 \,.
 \label{strima}
 \ee
As functions of the coupling
$G$
of the regularized centrifugal-like spike
these eigenvalues are sampled in Fig.~\ref{reone}.

%********** Figure 1 zde
\begin{figure}[h]                     %instead of \begin{figure}[t]
\begin{center}                         %instead of \begin{center}
\epsfig{file=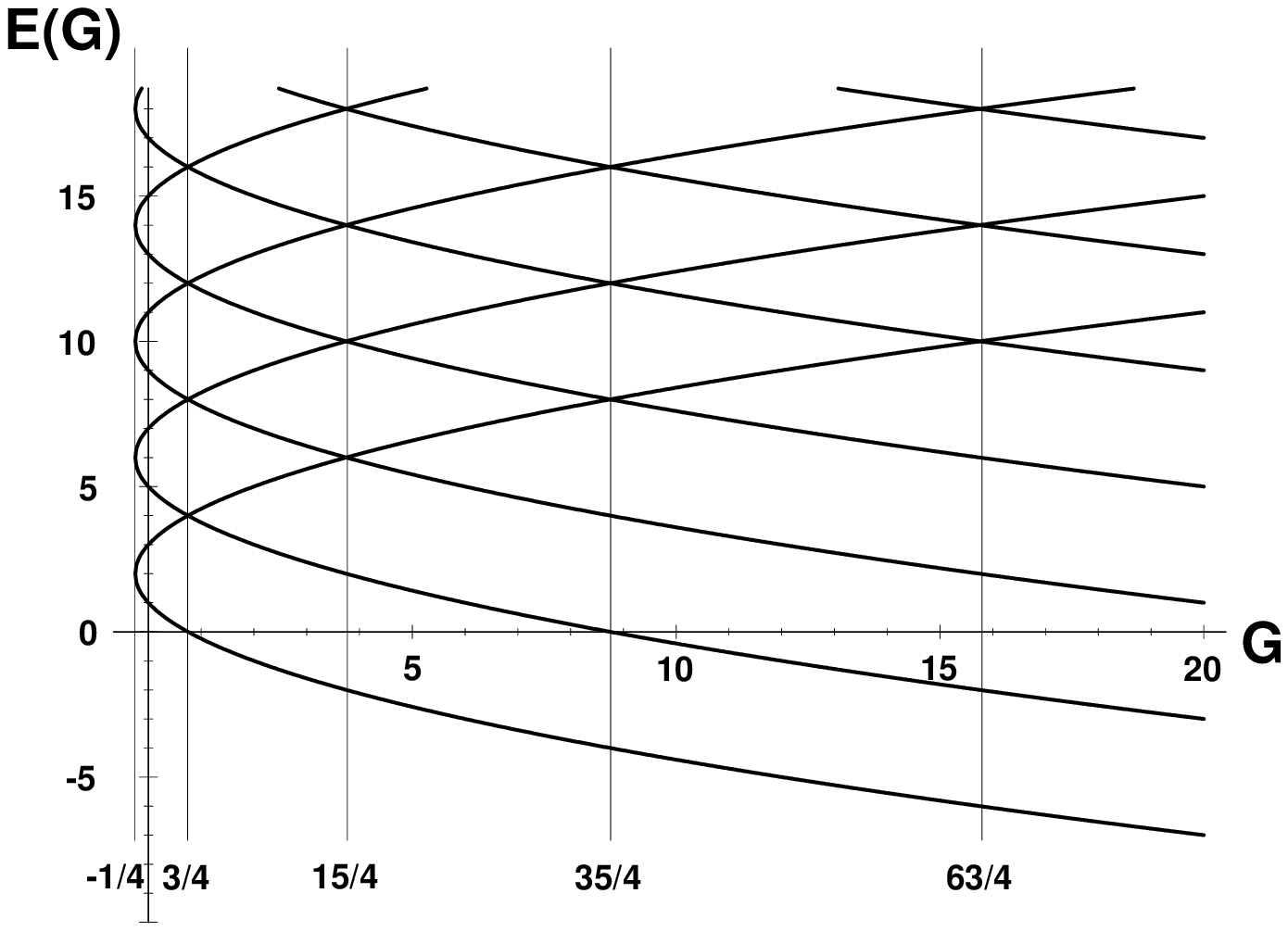,angle=270,width=0.72\textwidth}
\end{center}                         %instead of \end{center}
\vspace{-2mm} \caption{$G-$dependence of spectrum
of ${\cal PT}-$symmetric
harmonic oscillator (\ref{Halfa}).
Vertical lines mark the
exceptional-point values of coupling $G^{(EP)}_\alpha=\alpha^2-1/4$
at $\alpha=\alpha^{(EP)}=0,1,2,3$ and $4$.
 \label{reone}}
\end{figure}

In our present paper a  long missing
constructive
probabilistic interpretation
of such an exactly solvable quantum model
will be presented in a restriction to
the most interesting dynamical regimes
which are not too far from the instants of phase transitions
called exceptional points (EPs, \cite{Kato}).

\section*{Physics behind ${\cal PT}-$symmetric harmonic oscillator}

At the time of the publication of Ref.~\cite{ptho} in 1999, a
consistent physical unitary-evolution interpretation of similar
nonhermitian Hamiltonians has not been available yet. During the
first years of the new millennium people still preferred the
conventional phenomenological treatment of similar models, based on
the widely known Feshbach's effective-Hamiltonian philosophy
\cite{Feshbach,Nimrod}. Only step by step it has been clarified that
one has to distinguish, very strictly, between such a traditional,
``manifestly nonhermitian'' approach (in which the systems are, in
general, resonant or dissipative, and in which the
phenomenologically meaningful spectra of energies may be, and are,
in general, complex) and its ``hiddenly Hermitian'' alternative as
proposed by Bender with Boettcher (in \cite{BB} they insisted on the
strict reality and stable bound-state physical meaning of the energy
spectra).

In our present paper we will only pay attention to the latter branch
of the theory. Readers interested in the current status of the
former, wider branch of research may find its representative sample,
e.g., in the very recent edited book \cite{Christodoulides}. In
contrast, several extensive introductions into the latter, unitary
quantum theory using nonhermitian Hamiltonians with real spectra may
be found reviewed, e.g., in the pair of books \cite{Carlbook,book}.

In an application of the latter,
unitary-evolution approach to our present model (\ref{Halfa})
it is necessary to point out, first of all,
that the underlying, user-friendly Hilbert space ${\cal K}=L^2(\mathbb{R})$
loses the status of the physical space of states
with the conventional probabilistic interpretation.
For this reason,
one has to emphasize that also the variable $x$ still
{\em cannot\,} carry
the physical meaning of an observable of a particle position \cite{Batal}.
In this setting, the intuitively appealing concept of
$\,{\cal PT}-$symmetry must be assigned its
amended mathematical
meaning of a Krein-space-based
${\cal P}-$pseudo-Hermiticity of the Hamiltonian \cite{ali,Langer,AKbook}.

In a retrospective it is possible to say that several years were
needed for an ultimate correct reformulation of quantum mechanics in
which one comes to the conclusion that even the manifestly
nonhermitian Hamitonians with real spectra may still generate a
stable and unitary evolution. Incidentally, during the step-by-step
discoveries of the ultimate consistent formulation of the theory
(see, e.g., \cite{Carl,ali}) people were also rediscovering the
applicability of an older knowledge of the problem not only in
abstract mathematics \cite{Dieudonne} but also in several pragmatic
reinterpretations of the first principles of quantum mechanics by
physicists \cite{Dyson,Geyer}.

%Representing a quantum state

\subsection*{Three Hilbert space formulation of quantum mechanics\label{parappa}}

During the process of understanding of the Bender's and Boettcher's
conjectures \cite{BB} it appeared necessary to replace, first of
all, the mathematically friendly Hilbert space ${\cal K}$ by its
unitarity-compatible alternative (in a way recommended in
\cite{SIGMA} we will denote it by
dedicated symbol ${\cal H}$ in what follows). These two Hilbert
spaces differ just by na amendment of the inner product
$(\bullet,\bullet)$ \cite{Geyer}. In the standard Dirac's bra-ket
notation we may write
 \be
 (\psi_a,\psi_b)_{\cal K}=\br \psi_a|\psi_b \kt\,,
 \ \ \ \ \
 (\psi_a,\psi_b)_{\cal H}=\br \psi_a|\Theta|\psi_b \kt\,.
 \label{innpro}
 \ee
Here, the {\it ad hoc,},
Hilbert-space-metric operator
$\Theta$
must satisfy several compatibility conditions, thoroughly studied
and listed in \cite{Geyer}.
For our present purposes we only recall that this operator
must be self-adjoint in ${\cal K}$,
$\Theta=\Theta^\dagger$. This, as a consequence, guarantees the unitarity of the
evolution of the system in ${\cal H}$.
Moreover, one must also guarantee that this operator is bounded and invertible,
with bounded inverse \cite{Geyer}. Last but not least,
we need that the use of this Hilbert-space-metric
operator reinstalls the correct probabilistic contents of the
underlying quantum theory, i.e., that the condition $\Theta>0$
of its positive definiteness leads to the
standard norm in ${\cal H}$.

This means
that the evolution
generated by $H$ will
be unitary in  ${\cal H}$ \cite{Stone,[a2]}.
At the same time, this implies
that the metric must be,
by construction, Hamiltonian-dependent,
i.e., such that
 \be
 H^\dagger \,\Theta=\Theta\,H\,.
 \label{dieux}
 \ee
Fortunately, once we factorize
the metric
 \be
 \Theta=\Omega^\dagger\Omega
 \label{fakt}
 \ee
we reveal that all of the above requirements
are compatible, and that
the
operator $\Omega$
maps the ket-vector elements
$|\psi\kt$
of ${\cal K}$ (or, equivalently, of ${\cal H}$)
on the new, ``curly'' kets
which span another, third Hilbert space denoted as ${\cal L}$,
 \be
 |\psi\pkt = \Omega\,|\psi\kt\,,
 \ \ \ \ \  |\psi\kt \, \in \,{\cal H}\,,
 \ \ \ \ \ |\psi\pkt \, \in \,{\cal L}\,.
 \label{tato}
 \ee
This construction implies the equivalence between the
following two inner products,
 \be
 (\psi_a,\psi_b)_{\cal H}=(\psi_a,\psi_b)_{\cal L}\,.
 \ee
We may conclude that the quantum system
in question may be represented
by the Hamiltonian $H$ acting
in Hilbert space ${\cal H}$
or, equivalently,
by the Hamiltonian
 \be
 \mathfrak{h}=\Omega\,H\,\Omega^{-1}
 \ee
defined in Hilbert space ${\cal L}$. In this framework we may say
that $\mathfrak{h}$ is self-adjoint in ${\cal L}$ while $H$ is
self-adjoint in ${\cal H}$ \cite{[a2]}. The third, manifestly
unphysical Hilbert space ${\cal K}$ is just a mathematically
preferred auxiliary space in which the calculations are all
performed -- this is the reason why $H$ is often (and misleadingly)
called nonhermitian.

In an application of the three-Hilbert-space picture to our harmonic
oscillator Hamiltonian we may also observe that it is nonhermitian
in the manifestly unphysical Hilbert space ${\cal
K}=L^2(\mathbb{R})$. Obviously, unless we specify the physical
Hilbert space ${\cal H}$ [i.e., the metric $\Theta$], the
description of the system remains unfinished, leaving the
information about physics {\em incomplete}. The necessity of the
completion (i.e., of the specification of metric $\Theta$) follows
from the necessity of the standard probabilistic interpretation of
the model. In applications, such a requirement reflects the weakest
point of the whole theory. In fact, for a long time it remained
unnoticed that the present spiked harmonic oscillator model offers
one of the rare opportunities of its consequent and complete
implementation.

\subsection*{Exceptional points}

Besides
an expected confirmation of
complexification
of the whole spectrum
of model (\ref{Halfa})
at negative $\alpha<0$
(the effect widely known under the nickname of a
spontaneous breakdown of ${\cal PT}-$symmetry
\cite{BB,DDT}),
one of the key results of paper \cite{ptho} was
the observation that at the
positive integer values of $\alpha=1,2,\ldots$
the energy levels cross but remain real.
Due to the exact solvability of the model it was
easy to reveal that at all of these values of the parameter
marking the unavoided eigenvalue crossings
were accompanied by the
parallelization and degeneracy of the
related pairs of eigenvectors.
Indeed, for the  bound state wave functions
expressed in terms of Laguerre polynomials,
  \be
\varphi(x) = const. \,(x-ic)^{-Q \alpha+1/2}e^{-(x-ic)^2/2} \
L^{(-Q \alpha)}_n \left [
 (x-ic)^2
 \right ]\,,\ \ \ n=0,1,\ldots
\label{waves}
 \ee
the rigorous proof of the parallelizations
was based on the elementary
identities like
 $$
  L^{(-1)}_{n+1}\left [ (x-ic)^2 \right ]
  =-(x-ic)^2\,L^{(1)}_n\left [ (x-ic)^2 \right ]
  $$
etc.
Using the terminology as introduced by Kato \cite{Kato}
all of the integer values of $\alpha=0,1,\ldots$ may be
called exceptional
points (EPs). At these values, operator $H^{(\alpha)}$ ceases
to be diagonalizable. In the context of quantum mechanics this
has the following important  consequence (see the reasons, e.g., in \cite{ali}).

\begin{lemma} \cite{ptho}
Operator (\ref{Halfa}) may play the role of Hamiltonian
of a unitary quantum system only if $\alpha>0$ and
$\alpha \notin \mathbb{Z}$.
\end{lemma}

%$\hat{\lambda}(t)$
%\begin{widetext}
\begin{table}[h]
\caption{EP degeneracies
%of spectra. % (\ref{strima}).
 }
\vspace{0.5cm}
 \label{owe}
\centering {\small
\begin{tabular}{||c||c|c|c|c|c||}
\hline \hline
  $\alpha=$
  &$0$&${1}$&${2}$&
     $3$&$4$
     %&$-{1}/{2}$&${1}/{2}$&${3}/{2}$&
  %   ${5}/{2}$&${7}/{2}$
     \\
     \hline
  &&&&& \vspace{-0.5cm}\\
  \multicolumn{1}{||c||}{${{E_n^{(Q)}}(G)}$}
    &$(G=-{1}/{4})$&$(G={3}/{4})$&$(G={15}/{4})$&
   $(G={35}/{4})$&$(G={63}/{4})$
     \\
 \hline
 %
%  \multicolumn{1}{||c||}{\rm {\rm {\rm  (see Fig.~\ref{6ja4b}) } }}
%   &\multicolumn{1}{c}{{\rm  Im} ${{s}}_n(\beta,\tau) = 0$,}
%    &\multicolumn{1}{c||}{{\rm  Re} ${{s}}_n(\beta,\tau) >
%   0$}
%     \\
 \hline
  \vdots&&&&&\\
  -6&&&&&$E_0^{(+)}$\\
  -4&&&&$E_0^{(+)}$&\\
  -2&&&$E_0^{(+)}$&&$E_1^{(+)}$\\
  0&&$E_0^{(+)}$&&$E_1^{(+)}$&\\
  2&$E_0^{(+)}=E_0^{(-)}$&&$E_1^{(+)}$&&$E_2^{(+)}$\\
  4&&$E_0^{(-)}=E_1^{(+)}$&&$E_2^{(+)}$&\\
  6&$E_1^{(+)}=E_1^{(-)}$&&$E_0^{(-)}=E_2^{(+)}$&&$E_3^{(+)}$\\
  8&&$E_1^{(-)}=E_2^{(+)}$&&$E_0^{(-)}=E_3^{(+)}$&\\
  10&$E_2^{(+)}=E_2^{(-)}$&\vdots&$E_1^{(-)}=E_3^{(+)}$
  &\vdots&$E_1^{(-)}=E_4^{(+)}$\\
  \vdots&\vdots&&\vdots&&\vdots
  \\
 \hline \hline
\end{tabular}}
\end{table}
%\end{widetext}

 \noindent
The detailed nature of EP-related degeneracies can vary with
our choice of $\alpha^{(EP)}=\alpha^{(EP)}_K=K$ where $K =0,1,\ldots$
(see Table \ref{owe}).
At these points
the lost possibility of diagonalization of  $H^{(\alpha)}$
can only be replaced by its
canonical representation,
 \be
 H^{{(\alpha)}}_{}\, Q^{{(\alpha)}}_{}
 = Q^{{(\alpha)}}_{}\,{\cal J}^{{(\alpha)}}
 \,,\ \ \ \
 \alpha=\alpha^{{(EP)}}\in \mathbb{Z}\,.
 \label{Crealt}
 \ee
An optimal choice of the infinite-dimensional canonical
representative ${\cal J}^{{(\alpha)}}$ of the EP limit of the
Hamiltonian will be specified below. This choice will enable us to
treat the transition-matrix solutions $Q^{{(\alpha)}}_{}$ of
Eq.~(\ref{Crealt}) as a certain degenerate EP analogue of the set of
eigenvectors forming an unperturbed basis. In such a perspective our
recent experience with the EP-based perturbation theory will find
its new application as a tool of making, finally, the consistent and
constructive physical interpretation of our nonhermitian but unitary
harmonic oscillator quantum model near its EP singularities
complete.

Expectedly \cite{Geyer},
without an additional information about dynamics
there will be infinitely many such completions.
In the related literature, unfortunately, one rarely finds
a sufficiently nontrivial example of such a variability
of options. In our present paper such an example
is provided.

% \newpage quantum system

\section*{Results}

\subsection*{Physical Hilbert space of oscillator near the
spontaneous breakdown of ${\cal PT}-$symmetry ($\alpha^{(EP)}=0$)}

Let us initiate our analysis of oscillator
(\ref{Halfa})
in the dynamical regime of the smallest positive
parameters $\alpha$.
Only in the next section we will make the analysis
complete by extending it
to all of the EP neighborhoods of
$\alpha \approx K$ with $K=1,2, \ldots$.

Near the lowermost EP limit $\alpha \to 0^+$ an
inspection of Fig.~\ref{reone} reveals that
the full, infinite-dimensional Hilbert space
may be decomposed into a sequence of
two-dimensional subspaces ${\cal K}^{[2]}_{(n)}$,
 \be
 {\cal K}=\bigoplus_{n=0}^\infty\,{\cal K}^{[2]}_{(n)}\,.
 \label{nespeci}
 \ee
The vanishing$-\alpha$
loss of the diagonalizability of $H^{(\alpha)}$
may be best reflected by the choice of the
canonical representation
matrix ${\cal J}^{{(0)}}$
of Eq.~(\ref{Crealt})
in the following
block-diagonal-matrix
form of a direct sum of Jordan matrices,
 \be
 {\cal J}^{{(0)}}=J^{[2]}(2)\bigoplus
  J^{[2]}(6)\bigoplus  J^{[2]}(10)
 \bigoplus \ldots\,,
 \ \ \ \ \
 J^{[2]}(E)=
 \left (
 \begin{array}{cc}
 E&1\\
 0&E
 \ea
 \right )\,.
 \label{diresu}
 \ee
Having specified this matrix we
have to solve Eq.~(\ref{Crealt}) yielding the
infinite-dimensional
transition matrix.
The columns of this matrix
may then play the role of
an unperturbed basis in ${\cal K}$.
Such a construction generates, finally, a simplified
isospectral
zero-order representation
of our Hamiltonian,
 \be
 \mathfrak{H}^{(0)}(\alpha)=
 \left [ Q^{{(0)}}_{}
 \right ]^{-1}\,H^{{(\alpha)}}_{}\,
  Q^{{(0)}}_{}=
 {\cal J}^{{(0)}} +
 {\rm corrections}
 %\alpha\,
 % {\cal V}^{{(0)}}(\alpha) + {\cal O}(\alpha^2)
 \,,\ \ \ \
 0<\alpha\ll 1\,.
 \label{reCrealt}
 \ee
In other words
this means that our Hamiltonian will have the
infinite-dimensional
block-diagonal matrix structure,
 $$
 %\left [ Q^{{(0)}}_{}
% \right ]^{-1} H^{{(\alpha)}}_{}\,
%  Q^{{(0)}}_{}=
 \mathfrak{H}^{(0)}(\alpha)=
 \left(
 \begin{array}{cc|cc|cc}
 2&1&0&0&0&\ldots\\
 0&2&0&0&0&\ldots\\
 \hline
 0&0&6&1&0&\ldots\\
 0&0&0&6&0&\ldots\\
 \hline
 0&0&0&0&10&\ldots\\
 \vdots&\vdots&\vdots&\vdots&\ddots&\ddots
 \ea
 \right )\ +\ {\rm corrections}
 $$
where  the blocks are the two-by-two Jordan
matrices. It is now necessary to specify the
first-order correction term in (\ref{reCrealt}).

In the dynamical regime
of small
and positive $\alpha$s, {\it i.e.},
close to the leftmost EP instability at $\alpha^{(EP)}_0 = 0$
we have to describe the energies
as functions of the coupling constant $G$.
Indeed, once we set $G=G^{(EP)}+ \xi$ and once we rewrite
formula (\ref{strima}) for energies as a function of $\xi$,
we notice the qualitative difference
between the left and right vicinities of $\xi=0$.
As long as only the right, real-energy vicinity
with $\xi >0$ in
$E_n^{(\pm)}= 4n+2 \pm 2\,\sqrt{\xi}$
is of our present interest,
we know that at its EP boundary
the whole spectrum
degenerates
pairwise,
$\lim_{\alpha \to 0}\,E^{(\pm)}_{{n}}\to 4n+2$, $n=0,1,\ldots$.

Our recent experience with
corrections
to non-diagonalizable matrices \cite{admissible}
warns us against a naive expectation that
the correction term in (\ref{reCrealt})
should be of order ${\cal O}(\alpha)$.
An independent version of the same warning
came also from
Ref.~\cite{DDTsusy}
and/or from
inspection of Fig.~\ref{reone}.
We found that
the dominant, leading-order correction
appearing in Eq.~(\ref{reCrealt})
may be written in an apparently counterintuitive
but still remarkably elementary
explicit
form,
 \be
 %\left [ Q^{{(0)}}_{}
% \right ]^{-1}\,H^{{(\alpha)}}_{}\,
%  Q^{{(0)}}_{}=
 \mathfrak{H}^{(0)}(\alpha)=
 {\cal J}^{{(0)}} + \xi
 {\cal V}^{{(0)}}
 +
 %\alpha\,
 % {\cal V}^{{(0)}}(\alpha) + {\cal O}(\alpha^2)
 {\rm higher\ order\ corrections}
  \,,\ \ \ \ \xi={\cal O}(\alpha^2)
 \label{horeC}
 \ee
with elementary block-diagonal matrix of perturbations
 \be
 {\cal V}^{{(0)}}=
 \left [J^{[2]}(0)
 \right ]^T\,\bigoplus
 \left [J^{[2]}(0)
 \right ]^T\,\bigoplus
 \left [J^{[2]}(0)
 \right ]^T\,\bigoplus \ldots
 \,
 \label{thesa}
 \ee
where, the superscript $^T$ marks the matrix transposition.

The main consequence of these formulae is that
in every two-dimensional subspace  ${\cal K}^{[2]}_{(n)}$
we have a block-diagonalized leading-order Hamiltonian
 \be
 \mathfrak{H}^{(0)}(\alpha)
 \approx \mathfrak{H}^{(0)}_0(\alpha) =
 {H}^{[2]}_{(0)}(\xi) \,\bigoplus\,
 {H}^{[2]}_{(1)}(\xi) \,\bigoplus\,\ldots
 \label{r13}
 \ee
where
 \be
  {H}^{[2]}_{(n)}(\xi)=
 J^{[2]}(E_n^{(+)})+\xi\,
 \left [J^{[2]}(0)
 \right ]^T=
 \left (
 \begin{array}{cc}
 E_n^{(+)}&1\\
 \xi &E_n^{(+)}
 \ea
 \right )
 \,,\ \ \ \ \ E_n^{(+)}=\left .E_n^{(+)}\right |_{\alpha=0}
 =4n+2\,.
 \label{aprha}
 \ee
For the latter submatrices we can solve
the related time-independent Schr\"{o}dinger equations
in closed form,
 \be
 H^{[2]}_{(n)}(\xi)\,
  \left (\begin{array}{c}
 1\\
 \eta_\pm
 \ea
 \right )=
 \left (E_n^{(+)}+\eta_\pm
 \right )\, \left (\begin{array}{c}
 1\\
 \eta_\pm
 \ea
 \right )\,,
 \ \ \ \ \eta_\pm=\pm \sqrt{\xi}\,.
 \label{bham13}
 \ee
This has the following consequence.

\begin{lemma}
For approximate two by two matrix Hamiltonians (\ref{aprha})
the unfolding energies are real if and only if
the small parameter $\xi$ is non-negative, $\xi \geq 0$.
\end{lemma}

 \noindent
At non-negative $\xi$ we have $\xi=\alpha^2$ and $\eta_\pm=\pm
\alpha$ in (\ref{bham13}). The approximate Hamiltonian (\ref{aprha})
may be then made Hermitian along the lines outlined above. Via a
mere redefinition of inner products (\ref{innpro}) our unphysical
but mathematically optimal Hilbert space ${\cal K}^{[2]}_{(n)}$ is
converted into its correct physical alternative ${\cal
H}^{[2]}_{(n)}$.

\begin{lemma}\label{lemma3}
Metric operators $\Theta$
making Hamiltonian (\ref{aprha})
Hermitian (in
${\cal H}^{[2]}_{(n)}$)
read
 \be
 \Theta=\Theta^{[2]}_{(n)}(\alpha,b_n)=
 \left (
 \begin{array}{cc}
 \alpha &b_n\\
 b_n&1/\alpha
 \ea
 \right )\,
 \label{me22}
 \ee
and form a one-parametric family
numbered by a real variable $b_n$ such that
 $|b_n|<1$.
\end{lemma}

\begin{proof}
The Hermiticity of matrix $H^{[2]}_{(n)}(\xi)$ in the physical
Hilbert space ${\cal H}^{[2]}_{(n)}$ means that this matrix
satisfies condition (\ref{dieux}). This condition (written in ${\cal
K}^{[2]}_{(n)}$) may be perceived as a set of linear equations for
the matrix elements of the unknown matrix $\Theta$. This matrix must
be Hermitian and positive definite \cite{Geyer}. Under these
constraints, an easy algebra leads to the result.
\end{proof}

\begin{thm}\label{thm4}
At small $\alpha$
the infinite-dimensional matrix Hamiltonian
$\mathfrak{H}^{(0)}_0(\alpha)$
of Eq.~(\ref{r13})
becomes Hermitian in the {\it ad hoc\,} physical Hilbert space
 \be
 {\cal H}=\bigoplus_{n=0}^\infty\,{\cal H}^{[2]}_{(n)}
 \label{tenpr}
 \ee
whenever we introduce, in
${\cal K}=\bigoplus_n\,{\cal K}^{[2]}_{(n)}$,
one of the amended, nontrivial inner-product metrics
 \be
 \Theta=\bigoplus_{n=0}^\infty\,\Theta^{[2]}_{(n)}(\alpha,b_n)\,.
 \label{m14}
 \ee
The optional sequence of parameters
$b_n \in (-1,1)$ with  $n=0,1,\ldots$ is arbitrary.
\end{thm}
\begin{proof}
The infinite-dimensional matrix Hamiltonian (\ref{horeC}) must be
shown compatible with the Dieudonn\'{e}'s Hermiticity condition
(\ref{dieux}), but this follows from the block-diagonality of the
participating infinite-dimensional matrices, and from Lemma
\ref{lemma3}.
\end{proof}

 \noindent
We see that at sufficiently small parameters $\alpha>0$,
our ${\cal PT}-$symmetric harmonic-oscillator
Hamiltonian (\ref{Halfa})
defined in auxiliary, unphysical Hilbert space ${\cal K} =
L_2(-\infty,\infty)$ of Eq.~(\ref{nespeci})
acquires
the status of standard self-adjoint
generator of unitary evolution.
Nevertheless,
different choices of the sequence of parameters
$\{b_n\}$ define
phenomenologically non-equivalent quantum systems.
In any such a system
the observables must be represented
by operators $\Lambda$
which are self-adjoint
in
the
physical Hilbert space ${\cal H}$ of Eq.~(\ref{tenpr}).
Even for the block-diagonal subset
$\Lambda=\bigoplus_n \Lambda^{[2]}_{(n)}$
of observables the general form
of their admissible
submatrices
 \be
 \Lambda^{[2]}_{(n)}=
 \left (
 \begin{array}{cc}
 u &v\\
 y&z
 \ea
 \right )\,
 \label{ne22}
 \ee
remains $b_n-$dependent in general. Indeed, in a parallel to
Eq.~(\ref{dieux}) these submatrices must satisfy the
metric-dependent Hermiticity constraint
 \be
 \left [\Lambda^{[2]}_{(n)}\right ]^\dagger\,
 \Theta^{[2]}_{(n)}(\alpha,b_n)=
 \Theta^{[2]}_{(n)}(\alpha,b_n)\,
 \Lambda^{[2]}_{(n)}\,.
 \label{conns}
 \ee
\begin{lemma}\label{lemma5}
Condition (\ref{conns}) is satisfied if and only if we restrict
 \be
 y=y(b_n)= \alpha^2v+\alpha\,b_n\,(z-u) \,
 \ee
in (\ref{ne22}).
 \end{lemma}
 % \nonindent
In an elementary check, the latter
construction of observables reproduces
the initial
leading-order
Hamiltonian at $v=1$ and $u=z=0$.
It is also easy to verify that the most popular complementary
observable of charge \cite{Carl} is obtained at
$v=1/\alpha$ and $u=z=b_n=0$.

Marginally, let us add that
once we reparametrize $\alpha=\exp t$ and $b_n=\cos \phi$
in (\ref{me22}) and once we put $\phi=\mu+\nu$
we may also
factorize the metric
(cf. Eq.~(\ref{fakt})) yielding
 \be
 \Omega^{[2]}_{(n)}=
 \left (
 \begin{array}{cc}
 p&a\\
 a&q
 \ea
 \right )\,,\ \ \ \
 p=e^{t/2} \sin \mu\,,
 \ \ \ \
 q=e^{-t/2} \sin \nu\,,
 \ \ \ \
 a=e^{t/2} \cos \mu=e^{-t/2} \cos \nu\,.
 \label{tuje}
 \ee
On these grounds, whenever needed, we may perform transition to the
third Hilbert space ${\cal L}$ using Eq.~(\ref{tato}). Redundant as
this step may seem to be, the work in the latter space is often
recommended in conventional textbooks, mainly for establishing
easier contacts with experimentalists (cf., e.g.,
\cite{Batal,Bishop}).

\subsection*{Physical Hilbert spaces of oscillators near unavoided
level crossings ($\alpha^{(EP)}=1,2,\ldots$)}

In Fig.~\ref{reone}
we notice a significant qualitative difference
between the leftmost EP
at $\alpha=\alpha^{(EP)}_0 = 0$
(to the left of which the spectrum complexifies)
and the remaining EP family of
$\alpha^{(EP)}_K= K$ with $K=1,2,\ldots$
(in the
respective vicinities of which the spectra remain real).
In what follows we intend to show that  at $K \geq 1$
such a qualitative phenomenological difference
is also reflected by the related
mathematics.

First of all we notice that in the limit
$\alpha \to \alpha^{(EP)}_K$
the $K-$plet of
the lowermost energy levels remains
non-degenerate (ND).
In a small vicinity of $\alpha^{(EP)}_K$
the
$K-$dimensional
Hilbert space ${\cal K}^{[K]}_{ND}$
spanned by
the corresponding
wave functions
may be characterized, for this reason, by the unit-matrix
metric of textbooks, $\Theta^{[K]}_{ND}=I$. Hence, this subspace
may be treated as
equivalent to its two physical alternatives,
${\cal K}^{[K]}_{ND}\equiv {\cal H}^{[K]}_{ND}\equiv
{\cal H}^{[K]}_{ND}$.
For this reason the full Hilbert spaces
 \be
 {\cal K}=
 {\cal K}^{[K]}_{ND}\,
 \bigoplus\,{\cal K}^{[2]}_{(0)}
 \bigoplus\,{\cal K}^{[2]}_{(1)}
 \bigoplus\,\ldots\,
 ,\ \ \ \
 {\cal H}=
 {\cal H}^{[K]}_{ND}\,
 \bigoplus\,{\cal H}^{[2]}_{(0)}
 \bigoplus\,{\cal H}^{[2]}_{(1)}
 \bigoplus\,\ldots\,
 %{\cal K}^{[2]}_{(2)}\,
 \label{redenes}
 %\ee
 %{\cal K}^{[2]}_{(2)}\,
 \label{denespe}
 \ee
may be, for our present purpose
of the construction of its physics-representing amendment,
reduced to the respective relevant tilded subspaces
 \be
 \widetilde{\cal K}=\widetilde{\cal K}^{[2]}_{(0)}
 \bigoplus\,\widetilde{\cal K}^{[2]}_{(1)}
 \bigoplus\,\widetilde{\cal K}^{[2]}_{(2)}\,\bigoplus\,\ldots\,,
 \ \ \ \
 \widetilde{\cal H}=\widetilde{\cal H}^{[2]}_{(0)}
 \bigoplus\,\widetilde{\cal H}^{[2]}_{(1)}
 \bigoplus\,\widetilde{\cal H}^{[2]}_{(2)}\,\bigoplus\,\ldots
 \,
 .
 \label{ufdenes}
 \ee
In the small left and right vicinities of exceptional points
$\alpha^{(EP)} =K\geq 1$
the reduction of attention
will also involve
the omission, from our considerations,
of the trivial, diagonal-matrix sub-Hamiltonian
$H^{[K]}_{ND}$ such that
 $$
 {\left (H^{[K]}_{ND} \right )}_{jj}=E^{(+)}_j=-2K+2+4j\,,
  \ \ \ \ j=0,1,\ldots,K-1\,.
 $$
In the only relevant (i.e., in our notation,
in the tilded)
part (\ref{ufdenes}) of full spaces
the rest of the spectrum
remains doubly degenerate forming the sequence
sampled in Table \ref{owe},
 $$
 E^{(-)}_n=E^{(+)}_{n+K}=2K+2+4n\,, \ \ \ \ n=0,1,\ldots\,
 $$
This enables us
to establish, in three steps,
several $K>0$ parallels with the preceding
$K=0$ results.
In the first step we introduce the tilded version
of the canonical EP Hamiltonian,
 \be
 \widetilde{\cal J}^{{(K)}}=J^{[2]}(2K+2)\bigoplus
  J^{[2]}(2K+6)
  %\bigoplus  J^{[2]}(2K+10)
 \bigoplus \ldots\,.
 \label{rwesu}
 \ee
Up to the omission of the first $K$
non-degenerate levels this is a perfect $K>0$ analogue
of the $K=0$ EP Hamiltonian of Eq.~(\ref{diresu}).
In the second step we define the
infinite-dimensional tilded
transition matrices $\widetilde{Q}^{(K)}$
as solutions of a tilded version of Eq.~(\ref{Crealt}).
In the third step, as above, we finally
use these transition matrices
to define the unperturbed basis
(cf. \cite{admissible}).

In the vicinity of $\alpha^{(EP)}_K=K$,
as a result,
the tilded
$K>0$ analogue of
the simplified Hamiltonian of
Eq.~(\ref{horeC}) is obtained,
 \be
  \widetilde{\mathfrak{H}}^{(K)}(\alpha)=
 \left [ \widetilde{Q}^{{(K)}}_{}
 \right ]^{-1}\,\widetilde{H}^{{(\alpha)}}_{}\,
  \widetilde{Q}^{{(K)}}_{}=
 \widetilde{\cal J}^{{(K)}} + \delta^2
 \widetilde{\cal V}^{{(K)}}
 +
 %\alpha\,
 % {\cal V}^{{(0)}}(\alpha) + {\cal O}(\alpha^2)
 {\rm higher\ order\ corrections}\,.
 % \,,\ \
 \label{horeB}
 \ee
The matrix
of perturbations itself remains the same as above, $
\widetilde{\cal V}^{{(K)}}={\cal V}^{{(0)}} $
[cf. Eq.~(\ref{thesa}) above].
What is, nevertheless, different is the role of
the new small parameter
$\delta=\delta(\alpha)=\alpha-K$.
One of the reasons is that the
unfolded spectrum remains real
at both of its signs. Hence, the
approximate
leading-order tilded Hamiltonian
 \be
 %\mathfrak{H}^{(0)}(\alpha)
% \approx
 \widetilde{\mathfrak{H}}^{(K)}_0(\alpha) =
 \widetilde{H}^{[2]}_{(0)}[\delta(\alpha)] \,\bigoplus\,
 \widetilde{H}^{[2]}_{(1)}[\delta(\alpha)] \,\bigoplus\,\ldots
 \label{dr13}
 \ee
with
 \be
 \widetilde{H}^{[2]}_{(n)}(\delta)=
 J^{[2]}(E_n^{(-)})+\delta^2\,
 \left [J^{[2]}(0)
 \right ]^T=
 \left (
 \begin{array}{cc}
 E_n^{(-)}&1\\
 \delta^2 &E_n^{(-)}
 \ea
 \right )
 \,,\ \ \ \ E_n^{(-)}=\left .E_n^{(-)}\right |_{\alpha=K}
 =2K+4n+2\,
 \label{caprha}
 \ee
has different spectral properties
determined by the related Schr\"{o}dinger equation
 \be
 \widetilde{H}^{[2]}_{(n)}(\delta)\,
  \left (\begin{array}{c}
 1\\
 \pm \delta
 \ea
 \right )=
 \left (E_n^{(+)}\pm \delta
 \right )\, \left (\begin{array}{c}
 1\\
 \pm \delta
 \ea
 \right )\,.
 \label{bbham13}
 \ee
Still, many of the consequences remain similar.

\begin{lemma}\label{lemma3b}
Metric operators
making Hamiltonian (\ref{caprha})
Hermitian in
$\widetilde{\cal H}^{[2]}_{(n)}$
form a one-parametric family
 \be
 \widetilde{\Theta}^{[2]}_{(n)}(\delta,c_n)=
 \left (
 \begin{array}{cc}
 \delta &c_n\\
 c_n&1/\delta
 \ea
 \right )\,
 \label{ume22}
 \ee
where
$\delta=\delta(\alpha)=\alpha-K \neq 0$ is small,
and where $-1<c_n<1$.
\end{lemma}
\begin{proof}
The construction
is analogous to the one described in
the proof of
Lemma \ref{lemma3}.
\end{proof}

\begin{thm}\label{thm7}
Tilded Hamiltonian
$\widetilde{\mathfrak{H}}^{(K)}_0(\alpha)$
of Eq.~(\ref{dr13})
is Hermitian in any tilded physical Hilbert space
$\widetilde{\cal H}$
of Eq.~(\ref{ufdenes})
characterized by the
metric
 \be
 \widetilde{\Theta}=\bigoplus_{n=0}^\infty\,
 \widetilde{\Theta}^{[2]}_{(n)}[\delta(\alpha),c_n)\,
 \label{dm14}
 \ee
where all of the parameters
$c_n \in (-1,1)$ are variable.
\end{thm}
\begin{proof}
In comparison with Theorem \ref{thm4}
the only modification of the proof is that
now we
ignore the low-lying bound-state $K-$plets as
controlled by trivial
metric $\Theta^{[K]}_{ND}=I$.
Thus, at $K>0$ the proof remains analogous while
paying attention just to the
``tilded'' Hilbert-space subspaces.
\end{proof}

 \noindent
In the light of the closeness of parallels
between the $K=0$ and $K > 0$ EP-related scenarios
we leave the
last-step $K>0$ upgrade of the construction of
the admissible classes of
observables (\ref{ne22}) to interested readers.
For compensation let us add here that in general, the
variable parameters in metric (\ref{dm14})
may be chosen $\alpha-$dependent, $c_n=c_n(\alpha)$.
Fortunately, in the light of an appropriate
upgrade of Lemma \ref{lemma5} it is clear that
this would only imply an inessential
modification of the physics described by the model.

\section*{Discussion}

The core of our present message may be seen in the not quite
expected fact that the spiked and nonhermitian harmonic-oscillator
Hamiltonian (\ref{Halfa}) offers a truly exceptional sample of a
consequent application of the Bender-inspired, ${\cal
PT}-$symmetry-based reformulation (we called it
``three-Hilbert-space formulation'') of quantum mechanics of unitary
systems. In this sense and in the light of the above-mentioned
serious mathematical difficulties encountered during the study of
the (non-spiked) imaginary cubic anharmonic oscillator the present,
exactly solvable model could be assigned an important role of a new
benchmark in the ${\cal PT}-$symmetric quantum theory.

In the language of mathematics such a upgrade of the status
of the spiked
harmonic-oscillator model can be perceived as a consequence of the existence of
closed formulae (\ref{me22}) + (\ref{m14}) and  (\ref{ume22}) + (\ref{dm14}).
Near an arbitrary exceptional point of the system
they describe the {\em complete\,} set
of the metric operators and, hence, they determine
{\em all\,} of the
eligible, Hamiltonian-dependent
physical
Hilbert spaces ${\cal H}={\cal H}(H)$.
This is precisely the situation in which one
encounters an unrestricted possibility
of an exhaustive ``numbering'' of the
physical Hilbert spaces ${\cal H}(H)$
by the sets
of
parameters $\{b_n\}$ or $\{c_n\}$.

From the perspective of physics
one can speak about the {\em complete} menu
of the unitary quantum systems
possessing the standard probabilistic interpretation
and
compatible with
our preselected Hamiltonian $H$.
In a way explained in \cite{Geyer} (cf. also \cite{Carl,arabky}
for some elementary illustrative examples),
the ambiguity of this menu of spaces
can subsequently be restricted (or even completely suppressed)
by means of taking
some other candidates $\Lambda$ for the observables
into consideration.

In an opposite direction,
in a way proposed in \cite{Lotor})
one could also
recall the menu of metrics,
pick up one of them,
$\Theta_0$, and demand that any element $\Lambda_0$
of the class of eligible observables
must satisfy the constraint
 \be
 \Lambda_0^\dagger\,\Theta_0
 =\Theta_0\,\Lambda_0\,.
 \ee
This just means that the knowledge of the complete menu of metrics
enables us to keep also all of the admissible operators
representing
observables
self-adjoint in
our preselected Hilbert space ${\cal H}_0$.

In the latter context let us emphasize,
last but not least, that the full generality of
our present results concerning the special
spiked-oscillator model opens also a truly remarkable
possibility
of studying an
interplay between
the influence of changes of the separate
parameters. This could make the phenomenological
interpretation of our apparently not too complicated
model extremely flexible.
{\it Pars pro toto\,} let us mention
that one of the consequences
of this flexibility
(with a detailed analysis already lying far
beyond the scope of our present short paper)
might be sought near the
reality-of-levels preserving unavoided-level-crossing
interfaces with $\alpha^{(EP)}=1,2, \ldots$.
Indeed,
at these boundary points we could admit
{\em  discontinuous\,}
jumps in the matrix elements
$c_n=c_n(\delta)$ of the metric at $\delta=0$.
Every such a jump
(reflecting the ``punched'', two-sided
nature of the theoretically admissible
diagonalizable-Hamiltonian vicinity of the EPs)
would have to be interpreted
as a genuine quantum
catastrophe {\it alias},
in the terminology of Refs.~\cite{Denis,Denisb},
a quantum phase transition
of the second kind.

%\newpage

\section*{Acknowledgments}

The author acknowledges the support by the
 Faculty of Science of the University of Hradec
Kr\'{a}lov\'{e} and, in particular, by the
Excellence project 2212 P\v{r}F UHK 2020.

%\newpage

\section*{Data Availability}

 No datasets were generated or analysed during the current study.

\section*{Author Contribution statement}

The author is the single author of the paper.

\section*{Competing interests}

The author declares no competing interests.

%\newpage

\end{document}